\documentclass[11pt]{article}

\usepackage{amsmath, amsthm, amsfonts,amssymb}
\usepackage{graphicx}
\usepackage{color}
\usepackage{marvosym}
\usepackage{mathrsfs}
\usepackage{graphicx} 
\usepackage {csquotes}
\usepackage{soul}

\usepackage[affil-it]{authblk}

\newtheorem{theorem}{Theorem}

\newtheorem{lemma}[theorem]{Lemma}
\theoremstyle{definition}

\newtheorem{example}{Example}

\begin{document}
\title{Simpler proof for nonlinearity of majority function}
\author {Thomas W. Cusick
	\thanks{Email: \text{cusick@buffalo.edu}}}
\affil {Department of Mathematics, University at Buffalo, Buffalo, NY 14260, USA}

\date{}	
\maketitle
\begin{abstract}
Given a Boolean function $f$,  the (Hamming) weight $wt(f)$ and the nonlinearity $N(f)$ are well known
to be important in designing functions that are useful in cryptography.  The nonlinearity is expensive to compute, in general, so
any formulas giving the nonlinearity for particular functions $f$ are significant.  The well known majority function has been extensively 
studied in a cryptographic context for the last dozen years or so, and there is a formula for its nonlinearity.  The known proofs for this formula rely on many detailed results for the Krawtchouk polynomials.  This paper gives a much simpler proof.
\end{abstract}

\section{Introduction}
Boolean functions have many applications, particularly in coding theory and cryptography. A detailed account of the latter applications can be found in the book \cite{CS}. If we define $\mathbb{V}_n$ to be the vector space of dimension $n$ over the finite field $GF(2) = \{0, 1\},$ then an $n$ variable Boolean function 
$f(x_1, x_2, . . . , x_n) = f({\bf x})$ is a map from $\mathbb{V}_n$ to $GF(2).$ Every Boolean function $f({\bf x})$ has a unique polynomial representation (usually called the algebraic normal form \cite[p. 8]{CS}), and the degree of $f$ is the degree of this polynomial. A function of degree at most $1$ is called affine, and if the constant term is $0$ such a function is called linear. 
We let $B_n$ denote the set of all Boolean functions in $n$ variables, with addition and multiplication done mod $2.$
If we list the $2^n$ elements of $\mathbb{V}_n$ as $v_0 = (0,...,0), v_1 = (0,...,0,1), ...$ in lexicographic order (we abbreviate this as lexico order below), then the $2^n$-vector 
$(f(v_0), f(v_1), \ldots, f(v_{2^n-1}))$ is called the truth table of $f.$ The weight (also called Hamming weight) $wt(f)$ of $f$ is defined to be the number of $1's$ in the truth table for $f$. In many cryptographic uses of Boolean functions, it is important that the truth table of each function $f$ has an equal number of $0's$ and $1's$; in that case, we say that the function $f$ is balanced.

The distance $d(f,g)$ between two Boolean functions $f$ and $g$ in the same number of variables is defined by $$d(f, g) = wt(f + g),$$
where the polynomial addition is done mod $2.$ An important concept in cryptography is the nonlinearity $N(f)$ defined by 
$$N(f) = \underset{a ~affine~~~~~~}{\min d(f,a)}.$$ In order for a Boolean function to be useful in a cryptographic application, it is usually necessary that the function has high nonlinearity (see, for example, 
\cite[p. 122]{CS}).  So-called Fourier analysis of Boolean functions (see \cite[Chapter 2]{CS}) is very important in cryptography and other contexts.  The efficient computation of values of the nonlinearity is important here, and for this a very important tool is the Walsh transform of a Boolean function $f_n$ in $n$ variables. The Walsh transform of $f_n$ is the map $W_f: \mathbb{V}_n \rightarrow \mathbb{R}$ defined by
$$ W_f({\bf w}) = \sum_{x \in \mathbb{V}_n} (-1)^{f({\bf x}) +{\bf w} \cdot {\bf x}}, $$
where the values of $f= f_n$ are taken to be the real numbers $0$ and $1.$ This Walsh transform has also been important in physics and other sciences for at least $40$ years.  In this paper we want the Walsh transform because of the well known formula (see \cite[Th. 2.21, p. 17]{CS}, where $\hat{f}$ is used instead of $f$)
\begin{equation}
\label{NfW}
N(f_n) = 2^{n-1}- \frac{1}{2} \underset{{\bf u}\in \mathbb{V}_n} \max |W_f({\bf u})|.
\end{equation}
We shall use the obvious fact
$$W_f({\bf 0}) = 2^n - 2wt(f)$$
without comment in the rest of the paper.

\section {Weight = nonlinearity if weight is small enough}
It follows from the definitions that the nonlinearity is always $\leq$ the weight.  The following lemma gives a useful sufficient condition for the weight and nonlinearity to be equal. This lemma is well known to experts, but I have not found a statement of it in the literature.
The proof below was shown to me by Claude Carlet; it is much simpler than the proof using the Walsh transform given in \cite{C17}, which 
is a preliminary version of this paper. Before \cite{C17} was written, this lemma was stated in a version of \cite{Sang}. 

\begin{lemma}
\label{Bu}
Suppose $f$ is a Boolean function in $n$ variables with $wt(f) \leq 2^{n-2}.$ Then $wt(f)=N(f)$.
\end{lemma}
\begin{proof}
We give a proof by contradiction.  Suppose $a$ is an affine function such that
$$ d(a,f) < wt(f) \leq 2^{n-2}.$$
Then $a$ is not constant and by the triangle inequality ($\bf 0$ is the zero function)
$$wt(a)=d(a, {\bf 0}) \leq d(a,f) + d(f,{\bf 0})<2^{n-2}+2^{n-2} = 2^{n-1}.$$ 
This is a contradiction since the nonconstant affine function $a$ is balanced, and so has weight $2^{n-1}.$
\end{proof}

Lemma \ref{Bu} is sharp in the sense that it is no longer true if we replace the upper bound on $wt(f)$ by $2^{n-2}+1.$ For example, the function with $n=3$ having truth table $(0,0,0,1,0,0,1,1)$ (its algebraic normal form is $x_1x_2x_3 + x_1x_2 + x_2x_3$) has weight $3$ and nonlinearity $1.$   However, sometimes we can use Lemma \ref{Bu} to
evaluate the nonlinearity of functions $g_n$ in $n$ variables even if $wt(g_n) \neq N(g_n),$ and even if $wt(g_n) > 2^{n-2}.$  An example of 
this (Theorem \ref{B2n}) is given in the next section.

\section{Finding nonlinearity of the majority function }
\label{majfnc}
The Boolean majority function $M_k({\bf x})$ can be defined in several slightly different ways, but here we use the most common one, namely
$$M_k({\bf x}) = 1 \text{ if and only if } wt({\bf x}) \geq k/2.$$ The majority function seems to have been first mentioned in a cryptographic context in the 1991 book \cite[pp. 70-80]{Ding91}, where only the case of $k$ odd was considered.  The function has been of special interest in cryptography for the last dozen 
years or so because this function and variants of it can be proven to have optimal algebraic immunity. We do not need to explain this technical
concept here; interested readers can find a discussion of it in \cite[pp. 174-176]{CS}. A useful recent survey of the applications of majority functions to algebraic immunity is given in \cite{Carl13}.  The papers \cite{BP05, Dalai06, Fu09, SM07, Su15, ST} 
and others referenced in those papers all contain recent work on cryptographic applications of these functions. It is necessary to know the 
nonlinearity of the majority function to begin such studies.  

A frequently quoted determination of the nonlinearity of the majority function was given in \cite[Th. 3, p. 52]{Dalai06}, using many detailed results on the Krawtchouk polynomials (these date back to 1929; see \cite[pp. 150-154]{McWS} for a study of these polynomials and proofs of some of their properties). 
It is convenient to deal with the cases of odd and even $k$ separately, so we define:
\begin{equation}
\label{odd}
M_{2n+1}({\bf x}) = 1 \text{ if and only if } wt({\bf x}) \geq n+1
\end{equation}
and 
\begin{equation}
\label{even}
M_{2n}({\bf x}) = 1 \text{ if and only if } wt({\bf x}) \geq n.
\end{equation} 
For any $k,$ we let $M(k)$ denote the truth table of $M_k({\bf x}).$ We let $A(k)$ and $B(k)$ denote, respectively, the left and right halves of the truth table, so \[M(k) = A(k) B(k)\] (juxtaposition of the truth tables, each one being thought of simply as a bitstring). We give a much simpler proof for the value of the nonlinearity of the majority function in Theorem \ref{maj} below, but first we need three preliminary lemmas. We use the notation $C(f)$ for the complement of $f,$ that is, the function obtained by switching $0$ to $1$ and $1$ to $0$ for every entry in the truth table of $f.$ Similarly, if $S$ is a bitstring, then $C(S)$ is its complement. Also, $S^*$ denotes the bitstring $S$ in reverse order.
\begin{lemma}
\label{M2n}
We have $M(2n+1) = C(M(2n))^*~M(2n)$ for $n \geq 2.$
\end{lemma}
\begin{proof}
By \eqref{odd}, we can describe the two halves of $M(2n+1)$ by
\begin{equation}
\label{Aodd}
A(2n+1) = \{0 {\bf y}: wt({\bf y}) \geq n+1, {\bf y}  \text{ in lexico order}\} 
\end{equation}  
and
\begin{equation}
\label{Aeven}
B(2n+1) = \{1 {\bf y}: wt({\bf y}) \geq n,{\bf y}  \text{ in lexico order} \}. 
\end{equation}
Now \eqref{Aodd} implies
$$A(2n+1)^* = \{0 {\bf y}: wt({\bf y}) \geq n+1, {\bf y}  \text{ in reverse lexico order}\},$$
so \eqref{Aeven} gives
\begin{equation}
\label{AA*}
B(2n+1) = C(A(2n+1))^*.
\end{equation}
It follows from \eqref{even} and  \eqref{Aodd} that
\begin{align*}
A(2n+1) &= \{C(1 {\bf y}): wt({\bf y}) \leq n, {\bf y}  \text{ in lexico order}\} \\&= \{C({\bf y} 1): wt({\bf y}) \leq n, {\bf y}  
\text{ in reverse  lexico order}\}\\
&= C(M(2n))^*
\end{align*}
and then \eqref{AA*} gives $B(2n+1) = M(2n).$
\end{proof}

Below is an example, in which we use the notation $a_i$ for a string of $i$ symbols $a = 0$ or $1.$
\begin{example}
$M(5) = 0_710_3101_3 ~~0_3101_301_7$
\end{example}
Note that it follows from \eqref{Aodd} and \eqref{AA*} that $wt(M(2n+1)) = 2^{2n},$ so $M_{2n+1}({\bf x})$ is a
balanced function.
\begin{lemma}
\label{NAA*}
For any truth table $A,$ the nonlinearity of the function with truth table  $A~ C(A)^*$ is given by
$$ N(A~ C(A)^*) = 2N(A).$$
\end{lemma}
\begin{proof}
For any affine function $a$ we have
$$d(A, a) = d(C(A), C(a)) = d(C(A)^*, C(a)^*).$$
Hence the lemma follows from the fact that $aC(a)^*$ is an affine function whenever $a$ is (see the
well known Folklore Lemma \cite[Lemma 2.2, p. 8]{CS}).
\end{proof}
\begin{lemma}
\label{wtAodd}
We have $wt(A(2n+1)) = 2^{2n-1} - \frac{1}{2}\binom{2n}{n}$ for $n \geq 2.$
\end{lemma}
\begin{proof}
If we let ${\bf b}_{2n}$ denote a bitstring of length $2n,$ then \eqref{Aodd} gives
$$wt(A(2n+1)) = 2^{2n} - |\{{\bf b}_{2n}:~wt({\bf b}_{2n}) \leq n\}|= 2^{2n} - \sum_{j=0}^{n} \binom{2n}{j}$$
and now the Binomial Theorem implies the formula in the lemma.
\end{proof}
Our next example illustrates Lemmas \ref{M2n} and \ref{wtAodd}.
\begin{example} 
\label{ex2}
\begin{align*}
M(6) &= 0_710_3101_3 ~~0_3101_301_7~~0_3101_301_7~~01_{15} \\
A(7) &= 0_{15}1~~0_710_3101_3~~0_710_3101_3 ~~0_3101_301_7
\end{align*}
\end{example} 

Now the nonlinearity formulas for  $M_k({\bf x})$ are given in the following
theorem.
\begin{theorem}
\label{maj}
The nonlinearity of the majority function for $n \geq 2$ is given by 
\begin{equation}
\label{Nodd}
N(M(2n+1)) = 2wt(A(2n+1))= 2^{2n} - \binom{2n}{n}
\end{equation}
and
\begin{equation}
\label{Neven}
N(M(2n)) = \frac{1}{2}N(M(2n+1)) = 2^{2n-1} - \frac{1}{2} \binom{2n}{n}.
\end{equation}
\end{theorem}
\begin{proof}
The value $2^{2n} - \binom{2n}{n}$ for $M(2n+1)$ (in a different notation) was obtained via \eqref{NfW} 
in \cite[Th. 4.20, p. 74]{Ding91}. The proof uses only elementary properties of binomial coefficients but is
long enough so that we will not repeat it here.  Now Lemma \ref{wtAodd} completes the proof of \eqref{Nodd}.
Finally, Lemmas \ref{M2n} and \ref{NAA*} imply \eqref{Neven}.
\end{proof} 
Note that Lemma \ref{M2n} gives $A(2n+1) = C(M(2n))^*$ (see Example \ref{ex2}) and with \eqref{Neven} plus Lemma \ref{wtAodd} we obtain
$$N(A(2n+1)) = N(M(2n)) = wt(A(2n+1),$$
which is a weight = nonlinearity result not implied by Lemma \ref{Bu}.

We can use Lemma \ref{Bu} to obtain even more information about the majority function, for example 
we can determine $N(B(2n)),$ which is given in the next theorem. The next example, which illustrates how
the truth tables of $M(i)$ change as $i$ increases, may be helpful in reading the proof.
\begin{example} 
\begin{align*}
M(2n-1) &= A(2n-1)~~B(2n-1) \\
M(2n)~~~~~ &= A(2n-1)~~B(2n-1)~~B(2n)\\
M(2n+1) &= Q1(2n+1)~~ M(2n-1)~~M(2n-1)~~B(2n)
\end{align*}
\end{example} 

\begin{theorem}
\label{B2n}
The nonlinearity of $B(2n)$ for $n \geq 3$ is equal to
\begin{align}
\label{B2na}
wt(C(B(2n))^*) &= 2 \sum_{j=n+1}^{2n-2}\binom{2 n - 2}{j} +  \binom{2 n - 2}{n}\\ 
\label{B2nb}
&= \sum_{j=n+1}^{2n-2}\binom{2 n - 2}{j} + 2^{2n-3} - \frac{1}{2}\binom{2n-2}{n-1}.
\end{align}
\end{theorem}
\begin{proof}
Let $Q1(2n+1)$ denote the first quarter of the truth table of $M(2n+1).$ It follows from Lemma \ref{M2n} that
\begin{equation}
\label{Q1B2n}
Q1(2n+1) = C(B(2n))^*.
\end{equation}
If we let $Q1A(2n+1)$ and $Q1B(2n+1)$ denote, respectively, the left and right halves of the truth table for $Q1(2n+1)$   
and we let ${\bf b}_{2n-2}$ denote a bitstring of length $2n-2,$ then
$$Q1A(2n+1) = \{000 {\bf b}_{2n-2}: wt({\bf b}_{2n-2}) \geq n+1, {\bf b}_{2n-2}  \text{ in lexico order}\}$$
and
$$Q1B(2n+1) = \{001 {\bf b}_{2n-2}: wt({\bf b}_{2n-2}) \geq n, {\bf b}_{2n-2}  \text{ in lexico order}\}.$$ 
Therefore
\begin{equation}
\label{wtQ1A}
wt(Q1A(2n+1)) =  \sum_{j=n+1}^{2n-2}\binom{2 n - 2}{j} 
\end{equation}
and
\begin{equation}
\label{wtQ1B}
wt(Q1B(2n+1)) =  \sum_{j=n}^{2n-2}\binom{2 n - 2}{j}. 
\end{equation}
Now \eqref{B2na} follows from \eqref{Q1B2n}, \eqref{wtQ1A} and \eqref{wtQ1B}, and \eqref{B2nb} 
follows from \eqref{B2na} by elementary properties of binomial  coefficients.
 
Simple estimates using \eqref{B2nb} show that $wt(C(B(2n))^*) < 2^{2n-2}.$ Hence Lemma \ref{Bu}
 gives $wt(C(B(2n))^*) = N(C(B(2n))^*) = N(B(2n)).$
\end{proof}
The methods in Section \ref{majfnc} can be used to clarify various earlier results in the literature. For example,
it is easy to show that the functions $\phi_{2k}$ discussed in \cite[p. 106]{Dalai05} satisfy $\phi_{2k} = M(2k).$

\end{document}